\documentclass[aps,pra,preprint,superscriptaddress,longbibliography,nofootinbib]{revtex4-1}
\usepackage[utf8]{inputenc}
\usepackage{amsfonts}
\usepackage{amsmath}
\usepackage{amssymb}
\usepackage{amsthm}
\usepackage{mathtools}
\usepackage{enumerate}
\usepackage[shortlabels]{enumitem}
\usepackage{array}
\usepackage{bm}
\usepackage{graphicx}
\usepackage{nicefrac}
\usepackage{booktabs}
\usepackage[all]{xy}
\usepackage{hyperref}
\usepackage{bbm}
\usepackage{tikz}
\usetikzlibrary{arrows}
\usetikzlibrary{positioning}
\usepackage{cancel}
\usepackage{color}
\usepackage{MnSymbol}
\usepackage{multirow}

\newcommand\R{\mathbb{R}}

\newcommand\mc[1]{\mathcal{#1}}
\newcommand\cH{\mc{H}}
\newcommand\cK{\mc{K}}
\newcommand\cJ{\mc{J}}
\newcommand\LH{{\cL(\cH)}}
\newcommand\LHsa{\mc{L}(\cH)_\mathrm{sa}}
\newcommand\LHS{{\LHop}}
\newcommand\LAS{{\LAop}}
\newcommand\LBS{{\LBop}}
\newcommand\LHop{{\LH^\mathrm{op}}}
\newcommand\LK{\mc{L}(\cK)}

\newcommand\cA{\mc{A}}

\newcommand\cL{\mc{L}}

\newcommand\LA{{\mc{L}(\cH_A)}}
\newcommand\LAsa{{\mc{L}(\cH_A)_\mathrm{sa}}}
\newcommand\LAop{{\LA^\mathrm{op}}}
\newcommand\LB{{\mc{L}(\cH_B)}}
\newcommand\LBop{{\LB^\mathrm{op}}}
\newcommand\LAxB{\mc{L}(\cH_A \otimes \cH_B)}
\newcommand\LAmB{\mc{L}(\cH_A, \cH_B)}
\newcommand\CJ{{CJ}}
\newcommand\JC{{JC}}
\newcommand\JA{{\mc{J}(\cH_A)}}
\newcommand\JAsa{{\mc{J}(\cH_A)_\mathrm{sa}}}
\newcommand\JB{{\mc{J}(\cH_B)}}
\newcommand\JBsa{{\mc{J}(\cH_B)_\mathrm{sa}}}

\newcommand\JH{{\mc{J}(\cH)}}
\newcommand\JHsa{\mc{J}(\cH)_\mathrm{sa}}
\newcommand\JK{\mc{J}(\cK)}

\newcommand\tr{\mathrm{tr}}
\newcommand\trA{\tr_{\cH_A}}
\newcommand\trB{\tr_{\cH_B}}
\newcommand\ra{\rightarrow}

\newtheorem{theorem}{Theorem}
\newtheorem{lemma}{Lemma}

\newcommand{\mf}[1]{#1}
\newcommand{\cmf}[1]{#1}

\begin{document}

\title{Variations on the Choi-Jamio\l kowski isomorphism}
\author{Markus Frembs}
\email{m.frembs@griffith.edu.au}
\author{Eric G. Cavalcanti}
\email{e.cavalcanti@griffith.edu.au}
\affiliation{Centre for Quantum Dynamics, Griffith University,\\ Yugambeh Country, Gold Coast, QLD 4222, Australia}

\begin{abstract}
    We address various aspects of a widely used tool in quantum information theory: the Choi-Jamio\l kowski isomorphism [A. Jamio\l kowski, Rep. Math. Phys., 3, 275 (1972)]. We review different versions of the isomorphism, their properties and propose a unified description that combines them all. To this end, we identify the physical reason for the appearance of the (basis-dependent) operation of transposition in the isomorphism as used in Choi's theorem [M.-D. Choi, Lin. Alg. Appl., 10, 285 (1975)]. This requires a careful distinction between Jordan algebras and the different $C^*$-algebras they arise from, which are distinguished by their order of composition.
    
    Physically, the latter encodes a choice of time orientation in the respective algebras, which relates to a number of recent results, including a characterisation of quantum from more general non-signalling bipartite correlations [M. Frembs and A. D\"oring, arXiv:2204.11471] and a classification of bipartite entanglement [M. Frembs, arXiv:2207.00024]. 
\end{abstract}

\maketitle

\section{Choi and Jamio\l kowski - two isomorphisms, two theorems}\label{sec: two isos two thms}

Let $\mc{L}(\cH_A) = M_n(\mathbb{C})$ and $\mc{L}(\cH_B) = M_m(\mathbb{C})$ be two finite-dimensional matrix algebras over the complex numbers. The \emph{Choi-Jamio\l kowski isomorphism} \cite{Jamiolkowski1972,Choi1975} refers to an identification of linear operators in the tensor product algebra $\LA \otimes \LB \cong \cL(\cH_A \otimes \cH_B)$ and linear maps $\cL(\cH_A,\cH_B) := \{\phi: \LA \ra \LB \mid \phi \text{\ linear}\}$ between these algebras. This identification is not unique, a number of different versions exist, each with different properties. This section contains a brief review of the most common versions of the isomorphism.

The main part of this paper is Sec.~\ref{sec: II}, where we highlight the physical significance of the (relative) operator ordering between algebras $\LA$ and $\LB$ in the Choi-Jamio\l kowski isomorphism (Sec.~\ref{sec: Jordan vs C*-algebras}), give a reformulation of Choi's theorem that makes explicit the dependence on the choice of (relative) operator ordering (Thm.~\ref{thm: Choi revisited} in Sec.~\ref{sec: reformulations}), and based on that propose a unified version of the Choi-Jamio\l kowski isomorphism (Sec.~\ref{sec: variations on the CJI}).

\subsection{Jamio\l kowski's isomorphism and theorem}

Ref.~\cite{Jamiolkowski1972} defines 
a linear isomorphism $J: \mc{L}(\cH_A,\cH_B) \ra \mc{L}(\cH_A \otimes \cH_B)$ 
by
\begin{align}\label{eq: JI}
    \begin{split}
        \rho^J_\phi = J(\phi) &:= \sum_k e_k^* \otimes \phi(e_k)\; ,\\
        J^{-1}(\rho)(a) 
        &:= \trA[\rho(a\otimes \mathbbm{1}_B)]\; ,
    \end{split}
\end{align}
for all $a \in \LA$, where $\{e_k\}_k \in \mathrm{ONB}(\LA)$ is any orthonormal basis in $\LA$ with respect to the Hilbert-Schmidt inner product (see Eq.~(\ref{eq: HS inner product}) below).\footnote{Following notation in Ref.~\cite{Jamiolkowski1972} and related work \cite{FrembsDoering2019b,FrembsDoering2022a,FrembsDoering2022b,Frembs2022a}, we use $*$ instead of $\dagger$ for the (Hermitian) adjoint.} In particular, $J$ is \emph{basis-independent}: let $\{f_l\}_l \in \mathrm{ONB}(\LA)$ be another orthonormal basis in $\LA$, then\footnote{\mf{Here, $\sum_l f_lf^*_l$ and $\sum_k e_ke^*_k$ are resolutions of the identity in the inner product space $\LA$; in particular, note that Eq.~(\ref{eq: basis-independence JI}) involves no matrix multiplication.}}
\begin{align}\label{eq: basis-independence JI}
    \begin{split}
        \rho^J_\phi
        = \sum_k e_k^* \otimes \phi(e_k)
        &= \sum_k \big((\sum_l f_lf_l^*)e_k\big)^* \otimes \phi(e_k)\\
        &= \sum_k \sum_l \overline{(e_k,f_l)}_\LA f_l^* \otimes \phi(e_k)\\
        &= \sum_k \sum_l f_l^* \otimes (f_l,e_k)_\LA \phi(e_k)\\
        &= \sum_l f_l^* \otimes \phi\big((\sum_k e_ke_k^*)f_l\big)
        = \sum_l f_l^* \otimes \phi(f_l)\; .
    \end{split}
\end{align}
Moreover, the following theorems hold for $J$.

\begin{theorem}[de Pillis \cite{dePillis1967}]\label{thm: dePillis}
    Let $\phi \in \LAmB$. Then $\phi$ is Hermiticity-preserving, that is, $^*\phi(a) :=(*\circ \phi)(a) 
    = (\phi \circ *)(a) =: \phi^*(a)$ for all $a \in \LA$, if and only if $J(\phi)$ is Hermitian.
\end{theorem}

\begin{theorem}[Jamio\l kowski \cite{Jamiolkowski1972}]\label{thm: Jamiolkowski}
    Let $\phi \in \LAmB$.
    \begin{enumerate}
        \item[(a)] $\phi$ is positive if and only if $J(\phi)$ is \emph{positive on pure tensors (POPT)} \emph{(}or \emph{block positive, cf. \cite{PaulsenShultz2013,Kye2022})}, that is, $\tr[J(\phi)(a \otimes b)] \geq 0$ for all $a \in \LA_+$, $b \in \LB_+$. 
        \item[(b)] Let further $^*\phi = \phi^*$. 
        Then $\phi$ is trace-preserving if and only if $\trB[J(\phi)] = \mathbbm{1}_A$.
    \end{enumerate}
\end{theorem}

\subsection{Choi's isomorphism and theorem}

The most commonly used version of the Choi-Jamio\l kowski isomorphism is motivated by a theorem due to Choi \cite{Choi1975} (Thm.~\ref{thm: Choi}), who defines a map $C: \LAmB \ra \LAxB$ by
\begin{align}\label{eq: CI}
    \begin{split}
        \rho^C_\phi = C(\phi) &:= \sum_{ij} |i\rangle\langle j| \otimes \phi(|i\rangle\langle j|)\; ,\\
        C^{-1}(\rho)(a) 
        &:= \trA[\rho(a^T\otimes \mathbbm{1}_B)]\; ,
    \end{split}
\end{align}
for all $a \in \LA$, where $\{|i\rangle\langle j|\}_{ij}$ denotes the orthonormal basis in $\LA$ built from an orthonormal basis $\{|i\rangle\}_i \in \mathrm{ONB}(\cH_A)$, and $T$ denotes transposition in that basis.\footnote{This is sometimes written $\rho^C_\phi := (\mathrm{id} \otimes \phi)(|\Phi\rangle\langle\Phi|)$, where $|\Phi\rangle = \sum_i |i\rangle \otimes |i\rangle$ is an (unnormalised) maximally entangled state. \cmf{Note that if $\phi$ is trace-preserving, then $\tr[\rho^C_\phi] = \tr[\rho^J_\phi] = d_A$ (cf. Thm.~\ref{thm: Jamiolkowski}).}}
Unlike Eq.~(\ref{eq: JI}), Eq.~(\ref{eq: CI}) depends on the choice of basis, in particular, note that transposition (in $C^{-1}$) is a basis-dependent operation.\footnote{We emphasise that basis-dependence refers to the operator basis $\{|i\rangle\langle j|\}_{ij} \in \mathrm{ONB}(\LA)$, not the basis $\{|i\rangle\}_i \in \mathrm{ONB}(\cH_A)$. Necessary and sufficient conditions for the validity of Choi's theorem (Thm.~\ref{thm: Choi}) with respect to different choices of bases $\{e_i\}_i \in \mathrm{ONB}(\LA)$ are analysed in Ref.~\cite{PaulsenShultz2013,Kye2022}.}
Nevertheless, as a special case of Stinespring's theorem \cite{Stinespring1955}, Choi's theorem yields the following useful characterisation of completely positive maps.

\begin{theorem}[Choi \cite{Choi1975}]\label{thm: Choi}
    Let $\phi \in \LAmB$. Then $\phi$ is completely positive (CP) if and only if $\rho^C_\phi$ is positive. That is, $C$ restricts to a map $C|_{\LAmB_\mathrm{CP}}: \LAmB_\mathrm{CP} \ra \LAxB_+$.
\end{theorem}

\subsection{A basis-independent variant of the Choi-Jamio\l kowski isomorphism}

The basis dependence in Eq.~(\ref{eq: CI}) is unappealing and has prompted alternative definitions of the Choi-Jamio\l kowski isomorphism. Of course, one can just resort to Eq.~(\ref{eq: JI}), yet Thm.~\ref{thm: Choi} does not hold in this case: given a completely positive map $\phi$, $\rho^J_\phi$ is generally not positive (cf. Ref.~\cite{LeiferSpekkens2013,PaulsenShultz2013,Kye2022}). A formal trick to bypass this issue has been suggested in Refs.~\cite{AllenEtAl2017,BarrettLorenzOreshkov2020} (see also Ref.~\cite{Ozawa2014}), which define a map $\JC: \LAmB \ra \mc{L}(\cH^*_A \otimes \cH_B)$ by\footnote{Ref.~\cite{AllenEtAl2017} defines $a_{A^*} := \tr_{\cH_A}[\tau^\mathrm{id}_A a]$ for $\tau^\mathrm{id}_A := \sum_{ij} |i\rangle_{A^*}\langle j| \otimes |i\rangle_A\langle j|$; hence, $(a_A)_{mn} = (a_{A^*})_{nm}$ as matrices. In this way, the transposition in Eq.~(\ref{eq: CI}) is absorbed into the mapping between $\cH_A$ and its dual $\cH^*_A$.}
\begin{align}\label{eq: CJI* pedestrian}
    \begin{split}
        \rho^{\JC}_\phi = \JC(\phi) &:= \sum_{ij} |i\rangle_{A^*}\langle j| \otimes \phi(|i\rangle\langle j|)\; ,\\
        \JC^{-1}(\rho)(a) 
        &:= \tr_{\cH^*_{A}}[\rho(a_{A^*}\otimes \mathbbm{1}_B)]\; , \quad a_{A^*} := a^T \quad \quad \forall a \in \LA\; .
    \end{split}
\end{align}
Here, $\{|i\rangle_{A^*} \in \cH^*_A\}_i$ corresponds to the dual basis of $\{|i\rangle_A \in \cH_A\}_i$,\footnote{Recall that the dual basis $\{\langle j|_A\}_j$ of $\{|i\rangle_A\}_i$ is uniquely defined by the condition $\langle j|i \rangle_A = \delta_{ij}$.} yet the elements $|i\rangle_{A*}$ are thought of as vectors (rather than dual vectors) such that $\langle j|i\rangle_{A^*} = \langle i|j\rangle_A = \overline{\langle j|i\rangle}_A$, and thus
\begin{align*}\label{eq: C* to J*}
    \begin{split}
        \tr_{\cH^*_A}[\rho^{\JC}_\phi(a_{A^*} \otimes \mathbbm{1}_B)]
        &= \tr_{\cH^*_A}[(\sum_{ij}|i\rangle_{A^*}\langle j| \otimes \phi(|i\rangle\langle j|))(\sum_{mn} |m\rangle (a_{A^*})_{mn}\langle n| \otimes \mathbbm{1}_B)] \\
        &= \sum_{ijkmn} \langle k|i\rangle_{A^*} \langle j|m\rangle_{A^*} (a_{A^*})_{mn} \langle n|k\rangle_{A^*} \ \phi(|i\rangle\langle j|) \\
        &= \sum_{ijkmn} \langle k|n\rangle_A (a_{A})_{nm} \langle m|j\rangle_A \langle i|k\rangle_A \ \phi(|i\rangle\langle j|) \\
        &= \tr_{\cH_A}[(\sum_{ij}|j\rangle\langle i| \otimes \phi(|i\rangle\langle j|))(\sum_{mn} |m\rangle a_{mn}\langle n| \otimes \mathbbm{1}_B)]
        = \trA[\rho^J_\phi(a \otimes \mathbbm{1}_B)]\; .
    \end{split}
\end{align*}
Consequently, Eq.~(\ref{eq: CJI* pedestrian}) is basis-independent by Eq.~(\ref{eq: basis-independence JI}), and by comparison with Eq.~(\ref{eq: CI}), $\JC$ restricts to a map from completely positive maps $\phi$ to positive matrices $\rho^{\JC}_\phi \in \mc{L}(\cH^*_A \otimes \cH_B)_+$. 

Nevertheless, a core feature remains hidden in the details (cf. \cite{Araujo_CJI}).
In the next section, we explicitly identify the physical significance of the transposition in Eq.~(\ref{eq: CI}) with the choice of operator ordering in $\LA$ and $\LB$, which suggests a unified version of the Choi-Jamio\l kowski isomorphism that retains Thm.~\ref{thm: dePillis}, Thm.~\ref{thm: Jamiolkowski}, Thm.~\ref{thm: Choi}, and recovers Eq.~(\ref{eq: CJI* pedestrian}).

\section{A unified version of the Choi-Jamio\l kowski isomorphism}\label{sec: II}

\subsection{Background - Jordan algebras, complex structure and inner product}\label{sec: Jordan vs C*-algebras}

At the heart of the the different versions and properties of the Choi-Jamio\l kowski isomorphism lies the subtle distinction between
Jordan and $C^*$-algebras. We will only need the bare essentials here, and liberally refer to Ref.~\cite{Hanche-OlsenStormer1984JordanOA,AlfsenShultz1998a,AlfsenShultz1998b} for many more details.\\

\textbf{Jordan algebras.} The operator product on $\LH$ admits the decomposition,
\begin{equation}\label{eq: decomposition of operator product}
    ab = \frac{1}{2}\{a,b\} - \frac{1}{2}[a,b]\; ,
\end{equation}
where $[a,b] = ab-ba$ denotes the commutator and $\{a,b\} := ab+ba$ denotes the anti-commutator. $(\LH,[\cdot,\cdot])$ defines a Lie algebra while $\JH := (\LH,\{\cdot,\cdot\})$ defines a Jordan algebra.\footnote{A Jordan algebra $(J,\circ)$ 
is commutative (but generally non-associative) and satisfies the Jordan identity: $(a\circ b)\circ(a\circ a) = a\circ(b\circ(a\circ a))$ for all $a,b \in J$. For a classification of Jordan algebras, see Ref.~\cite{JordanVNWigner1934,Hanche-OlsenStormer1984JordanOA,McCrimmon2004}.}\footnote{Overloading notation slightly, we write $\LH$ for the linear space and for the algebra of operators on $\cH$. \label{fn: overload LH}} Importantly, already $\JH_\mathrm{sa}$ defines a real Jordan algebra, in particular, the anti-commutator closes on self-adjoint (Hermitian) elements---unlike the operator product in $\LH$. Indeed, this was the original motivation to study Jordan algebras in Ref.~\cite{Jordan1933,JordanVNWigner1934}, viz. as generalisations of observable algebras in quantum mechanics.

Not every real Jordan algebra $\cJ$ arises as a self-adjoint part (closed under the anti-commutator) from a $C^*$-algebra $\cA$, i.e., $\cJ \not\subseteq (\cA_\mathrm{sa},\{\cdot,\cdot\})$ in general.\textsuperscript{\ref{fn: overload LH}} If it does, $\cJ$ is called \emph{special}, otherwise it is called an \emph{exceptional} Jordan algebra. The key to identify special Jordan algebras is the double role played by self-adjoint operators in $\LH$: as observables and as generators of symmetries \cite{GrginPetersen1974,AlfsenShultz1998b}.\footnote{Jordan algebras whose elements mimic this dichotomy are characterised by the existence of a \emph{dynamical correspondence} \cite{AlfsenShultz1998a,AlfsenShultz1998b} (see also Ref.~\cite{Baez2022}). In turn, for every dynamical correspondence on $\cJ$, there is an associative $C^*$-algebra $\cA$ from which $\cJ$ arises as a self-adjoint part $\cJ \subseteq (\cA_\mathrm{sa},\{\cdot,\cdot\})$ (Thm.~23 in Ref.~\cite{AlfsenShultz1998a}).\label{fn: dynamical correspondences}} In the case of $\JHsa$, there are exactly two such associative products that extend the anti-commutator in $\JHsa$ (cf. \cite{AlfsenShultz1998a}),
\begin{equation}\label{eq: two associative products}
    a \cdot_+ b := \frac{1}{2}\{a,b\} + \frac{1}{2}[a,b] = ab \quad \quad \quad a \cdot_- b := \frac{1}{2}\{a,b\} - \frac{1}{2}[a,b] = ba\; ,
\end{equation}
for all $a,b \in \JH$. In other words, $\JHsa$ arises as the self-adjoint part of both $\LH = (\LH,\cdot_+)$ and its opposite algebra \mf{$\LHop = (\LH,\cdot_-)$}, 
for which the order of composition is reversed with respect to $\LH$.\textsuperscript{\ref{fn: overload LH}} Since transposition $T$ is an order-reversing involution on $\LH$, 
it maps between $\LH$ and $\LHop$. We will see that the distinction between associative $C^*$-algebras $\LH$ and $\LHop$ is at the heart of the different versions of the Choi-Jamio\l kowski isomorphism and in turn of Choi's theorem (see Thm.~\ref{thm: Choi revisited} below). Before doing so, we briefly expand on the physical meaning of this distinction.\\

\textbf{Time orientations and complex structure.} Note that the two (associative) products in Eq.~(\ref{eq: two associative products}) are readily distinguished from one another by the sign in front of the commutator. This sign obtains a physical interpretation as the \emph{direction of time} as follows (cf. Ref.~\cite{AlfsenShultz1998a,FrembsDoering2019b,Doering2014,FrembsDoering2022a,Frembs2022a}). When described by $\LH$, the evolution of the system can be expressed via the map
\begin{equation}\label{eq: canonical time orientation}
    \Psi(t,a)(b) = e^{ita}be^{-ta} \quad \quad \forall a,b \in \JHsa, t\in\R\; . 
\end{equation}

\vspace{1.5cm}

Note that $\Psi$ describes unitary evolution in $\LH$ (as opposed to anti-unitary evolution \cite{Wigner_GruppentheorieUndQM,Bargmann1964}). With Ref.~\cite{FrembsDoering2022a,FrembsDoering2019b,Frembs2022a}, we therefore call it the \emph{canonical time orientation on $\LH$}.
The relation to the choice of sign in Eq.~(\ref{eq: decomposition of operator product}) follows by differentiation: for all $a \in \JHsa$
\begin{align}\label{eq: commutator as infinitesimal}
    \frac{d}{dt}\big|_{t=0}\Psi(t,a) = i[a,\cdot]\; .
\end{align}

This suggests a slight reformulation of the decomposition in Eq.~(\ref{eq: decomposition of operator product}) into
\begin{equation}\label{eq: complex structure on LH}
    \delta_a(b) := a\circ b - ia \star b\; ,
\end{equation}
where $a\circ b := \frac{1}{2}\{a,b\}$ and $a \star b := \frac{i}{2}[a,b]$ for all $a,b \in \JH$. Indeed, Eq.~(\ref{eq: complex structure on LH}) describes the natural decomposition of the \emph{order derivation}\footnote{An order derivation $\delta$ on a Jordan algebra $\cJ$ is a bounded linear operator such that $e^{t\delta}(b) \in \cJ_+$ for all $b \in \cJ_+$ and $t \in \mathbb{R}$ \cite{AlfsenShultz1998a}.} $\delta_a: \JH_+ \ra \JH_+$, $a \in \JH$ given by $\delta_a(x) := \frac{1}{2}(ax+xa^*) = \delta_a(x) + \delta_{ia}(x)$, into a \emph{self-adjoint} part $\delta_a$, $a \in \JHsa$ and a \emph{skew} part $\delta_{ia}$, $a \in \JHsa$ (Lm.~11 in Ref.~\cite{AlfsenShultz1998a}).\footnote{Note that while the operator product in $\LH$ 
does not close on self-adjoint elements, since the commutator of two Hermitian operators is skew-Hermitian, the product $\star$ \emph{does} close on self-adjoint elements.}
Moreover, it turns out that skew order derivations 
induce Jordan homomorphisms of the form $\Psi(t,a) = e^{2t\delta_{ia}}$ (Lm.~9 and Lm.~14 in Ref.~\cite{AlfsenShultz1998a}).

Importantly, Eq.~(\ref{eq: complex structure on LH}) defines a complex structure on $\JH$---viewed as the space of order derivations on $\JHsa$ \cite{Connes1974}---\mf{and it follows with Eq.~(\ref{eq: two associative products}) and Eq.~(\ref{eq: commutator as infinitesimal}) that such a complex structure promotes the Jordan algebra $\JH$ to the $C^*$-algebra $\LH$.\footnote{A  complex structure on the order derivations of a Jordan algebra is also called a \emph{Connes orientation} \cite{Connes1974}.}
Furthermore, we find that under this identification the `complex conjugation' map $\LH = \LHsa + i\LHsa \leftrightarrow \LH^* := \LHsa - i\LHsa$ is given by the Hermitian adjoint $*$ on $\LH$ (cf. Prop.~15 in Ref.~\cite{AlfsenShultz1998a}).}
Clearly, $*$ is an order-reversing involution on $\LH$, and by comparing Eq.~(\ref{eq: decomposition of operator product}), Eq.~(\ref{eq: commutator as infinitesimal}) and Eq.~(\ref{eq: complex structure on LH}) it induces a change in time orientation on $\JHsa$.\footnote{This central observation is at the heart of Thm.~3 in Ref.~\cite{FrembsDoering2022a} and Thm.~2 in Ref.~\cite{Frembs2022a}.} To emphasise the relation with (skew) order derivations and time orientations, below we will use the Hermitian adjoint as order-reversing map between $C^*$-algebras $\LH$ and $\LHop \cong \LH^*$.\\

\textbf{Inner product and Born rule.} The Hilbert-Schmidt inner product on $\LH$, 
\begin{equation}\label{eq: HS inner product}
    (a,b)_{\LH}
    := \tr_{\cH}[b^*a]\; ,
\end{equation}
can be used to succinctly encode Eq.~(\ref{eq: JI}) as follows (cf. \cite{Jamiolkowski1972}): for all $a \in \LA$, $b \in \LB$,
\begin{equation}\label{eq: CJI}
    (\rho^J_\phi,a^*\otimes b)_{\LA \otimes \LB}
    = \tr_{\cH_A \otimes \cH_B}[\rho^J_\phi(a\otimes b^*)] = \trB[\phi(a)b^*]
    = (\phi(a),b)_\LB\; .
\end{equation}
Note that (for a single system) Eq.~(\ref{eq: HS inner product}) is a property of $\JH$ only,
\begin{equation}\label{eq: HS inner product - Jordan algebras}
    (a,b)_\LH
    = \tr_{\cH}[b^*a] 
    = \frac{1}{2}\tr_\cH[\{b^*,a\}]\; ,
\end{equation}
since $\tr[ab-ba] = 0$ for all $a,b \in \LH$. The Hilbert-Schmidt inner product therefore restricts to the (complexified) Jordan algebra $\JH$. What is more, restricting to the real part $\JH_\mathrm{sa}$, Eq.~(\ref{eq: HS inner product - Jordan algebras}) is readily identified with the Born rule of the system described by $\JHsa$. In particular, note that we obtain the same Born rule from the inner product on $\LHS$,
\begin{equation}\label{eq: inner product in LH*}
    (a,b)_\LHS
    := \tr_{\cH^*}[b \cdot_- a^*]
    = \frac{1}{2}\tr_{\cH^*}[\{b,a^*\}]
    = (\tr_\cH[b^* \cdot_+ a])^*
    = \overline{(a,b)}_\LH\; . 
\end{equation}

Together with Eq.~(\ref{eq: CJI}) this suggest to study the Choi-Jamio\l kowski isomorphism under restriction to Jordan algebras.\footnote{\cmf{While we do not pursue this here, we remark that this raises the question of defining a tensor product of (special) Jordan algebras. This is naturally done in reference to their ambient complex $C^*$-algebras \cite{Hanche-Olsen1983} (see also Ref.~\cite{GrginPetersen1976}), which requires a choice of dynamical correspondence \cite{AlfsenShultz1998a}, equivalently time orientation in addition to the mere Jordan algebra structure. Importantly, this choice is not arbitrary, i.e., the Born rule of a composite quantum system \emph{does} depend on the relative choice of $C^*$-algebras \cite{FrembsDoering2022a,FrembsDoering2022b,Frembs2022a}.}} In the next section, we will show that Thm.~\ref{thm: dePillis} and Thm.~\ref{thm: Jamiolkowski} obtain under the restriction to Jordan algebras. In contrast, 
we will see that Thm.~\ref{thm: Choi} \emph{does} depend on the choice of complex structure \cite{Connes1974} (equivalently, dynamical correspondence \cite{AlfsenShultz1998a} or time orientation \cite{FrembsDoering2022a,Frembs2022a}) and thus on the choice of $C^*$-algebra $\LH$ or $\LHS$.




\subsection{Hermiticity-preserving, positive, decomposable and completely positive maps}\label{sec: reformulations}

In order to expose the significance of the choice of operator ordering on $\JA$ and $\JB$ in Thm.~\ref{thm: dePillis}, Thm.~\ref{thm: Jamiolkowski} and Thm.~\ref{thm: Choi}, we study the (linear) isomorphism $J$ as defined in Eq.~(\ref{eq: CJI}) under the restriction to Jordan algebras for different types of maps $\phi: \JA \ra \JB$. Note that $\LA \cong \JA$ ($\LAsa \cong \JAsa)$ as complex (real) linear spaces (similarly for $\LB$).\\

\textbf{De Pillis's theorem and Jamio\l kowski's theorem revisited.} We start with Hermiticity-preserving maps, $\phi: \JA \ra \JB$ such that $^*\phi = * \circ \phi = \phi \circ *= \phi^*$. 
This is equivalent to the condition that $\phi$ preserves the self-adjoint parts of $\JA$ and $\JB$, i.e., that $\phi$ restricts to a linear map $\phi|_{\JAsa}: \JAsa \ra \JBsa$.\footnote{Conversely, every linear map $\phi: \JAsa \ra \JBsa$ has a unique linear extension to the complexification $\JA := \JAsa +i\JAsa$, 
given by $\phi(a + ib) := \phi(a) +i\phi(b)$. The same is true for Hermiticity-preserving maps with additional structure (see Thm.~\ref{thm: dePillis for Jordan algebras} and Thm.~\ref{thm: Jamiolkowski for Jordan algebras}).\label{fn: maps between Jordan algebras}}

It follows that the set of operators $\rho_\phi^J$ corresponding to Hermiticity-preserving maps $\phi: \LA \ra \LB$ under Eq.~(\ref{eq: JI}) is the same as the set of  operators $\rho_\phi^J$ corresponding to linear maps $\phi: \JAsa \ra \JBsa$. Thm.~\ref{thm: dePillis} thus applies to (complex) Jordan algebras:

\begin{theorem}[de Pillis for Jordan algebras \cite{dePillis1967}]\label{thm: dePillis for Jordan algebras}
    Let $\phi: \JA \ra \JB$. Then $\phi$ restricts to a map $\phi|_{\JAsa}: \JAsa \ra \JBsa$ if and only if $J(\phi)$ is Hermitian.
\end{theorem}

Next, recall that a map $\phi: \LA \ra \LB$ is positive if and only if $\phi(a) \in \LB_+$ whenever $a \in \LA_+$, and that a positive operator is in particular self-adjoint. Hence, $\LA_+ \subset \JAsa$ (as real linear spaces) such that the identification between operators $\rho_\phi^J \in \LAxB$ that are positive on pure tensors and positive maps $\phi \in \LAmB_+$ also holds on the level of Jordan algebras:

\begin{theorem}[Jamio\l kowski for Jordan algebras \cite{Jamiolkowski1972}]\label{thm: Jamiolkowski for Jordan algebras}
    Let $\phi: \JAsa \ra \JBsa$.
    \begin{enumerate}
        \item[(a)] $\phi$ is positive if and only if $J(\phi)$ is positive on pure tensors (POPT).
        \item[(b)] Let further $^*\phi = \phi^*$. 
        Then $\phi$ is trace-preserving if and only if $\trB[J(\phi)] = \mathbbm{1}_A$.
    \end{enumerate}
\end{theorem}

\textbf{Decomposable and completely positve maps.} We can study the isomorphism under the restriction to Jordan algebras for maps with even more structure. Recall that a Jordan $*$-homomorphism $\Phi: \JA \ra \JB$ is a map such that $^*\Phi = * \circ \Phi = \Phi \circ * = \Phi^*$, and which preserves the anti-commutator in the respective algebras, i.e., $\Phi(\{a,a'\}) = \{\Phi(a),\Phi(a')\}$ for all $a,a' \in \JA$. More generally, a map $\phi: \JA \ra \JB$ is called decomposable, denoted by $\phi \in \mc{L}(\cH_A, \cH_B)_\mathrm{D}$, if it is of the form $\phi = v^* \Phi v$, where $v: \cH_B \ra \cK$ is linear and $\Phi: \JA \ra \JK$ is a Jordan $*$-homomorphism \cite{Stormer1982}. Since decomposable maps are by definition maps between Jordan algebras, we can define the corresponding set of operators $\rho_\phi^J$ (under $J$ in Eq.~(\ref{eq: JI})), denoted by $\mc{L}(\cH_A \otimes \cH_B)_\mathrm{D}$.

What about completely positive (CP) maps? Can we capture the set of CP maps (and thus the set of operators under Eq.~(\ref{eq: CI})) by studying their restriction to real
Jordan algebras? This is not the case! To see this, it is helpful to consider the classification of completely positive maps due to Stinespring \cite{Stinespring1955}: every completely positive map $\phi: \LA \ra \LB$ is of the form $\phi = v^*\Phi v$, where $v: \cH_B \ra \cK$ is linear and $\Phi: \LA \ra \LK$ is a $C^*$-algebra homomorphism, i.e., $^*\Phi = \Phi^*$ and $\Phi(aa') = \Phi(a)\Phi(a')$ for all $a,a' \in \LA$.

Of course, every $C^*$-algebra homomorphism restricts to a Jordan $*$-homomorphism $\Phi|_{\JA}: \JA \ra \JK$, however, the converse is generally not true: a Jordan $*$-homomorphism $\Phi: \JA \ra \JK$ only preserves commutators up to sign \cite{AlfsenShultz1998a}. More precisely, either $\Phi(aa') = \Phi(a)\Phi(a')$ or $\Phi(aa') = \Phi(a')\Phi(a)$, equivalently either $\Phi: \LA \ra \LK$ or $\Phi: \LAop \ra \LK$ is a $C^*$-algebra homomorphism \cite{Kadison1951}.

Recall (from Eq.~(\ref{eq: two associative products})) that both $\LA$ and $\LA^\mathrm{op}$ reduce to the same Jordan algebra $\JAsa$. It thus follows from Stinespring's theorem that restricting the Choi-Jamio\l kowski isomorphism to Jordan algebras involves both CP maps $\phi = v^*\Phi v: \LA \ra \LB$ and CP maps $\phi = v^*\Phi v: \LAop \ra \LB$. Evidently, a CP map $\phi: \LA \ra \LB$ is generally not also completely positive as a map $\phi: \LAop \ra \LB$,\footnote{Analysing when this is the case 
is closely related to Peres' separability criterion 
\cite{Peres1996,Horodeckisz1996} (see also Ref.~\cite{Frembs2022a}).} As an immediate consequence, we deduce that set of operators $\rho_\phi^J \in \mc{L}(\cH_A \otimes \cH_B)_\mathrm{D}$ (corresponding to decomposable maps) is strictly larger than the set of operators $\rho_\phi^J$ corresponding to CP maps, since
\begin{equation*}\label{eq: cp and dec maps}
    \mc{L}(\cH_A, \cH_B)_\mathrm{CP} 
    \subsetneq \mc{L}(\cH_A, \cH_B)_\mathrm{D}\; .
\end{equation*}
The above argument shows that Choi's theorem cannot be restricted to Jordan algebras as in the case of Thm.~\ref{thm: dePillis for Jordan algebras} and Thm.~\ref{thm: Jamiolkowski for Jordan algebras}; in essence: \emph{positivity of (quantum) states $\rho \in \mc{L}(\cH_A \otimes \cH_B)_+$ (equivalently by Thm.~\ref{thm: Choi}, complete positivity for maps $C^{-1}(\rho)$) depends on the (relative) operator ordering between the algebras $\LA$ and $\LB$} (cf. \cite{FrembsDoering2022a,FrembsDoering2022b,Frembs2022a}).

\vspace{0.4cm}

\textbf{Choi's theorem revisited.} The distinction between $C^*$-algebras in Sec.~\ref{sec: Jordan vs C*-algebras} exposes why Choi's theorem (Thm.~\ref{thm: Choi}) does not hold for the isomorphism $J$ in Eq.~(\ref{eq: JI}):\footnote{That is, given a completely positive map $\phi: \LA \ra \LB$, $\rho_\phi^J \in \LAxB$ is generally not positive.}
\emph{$J$ and $C$ are implicitly defined with respect to different 
$C^*$-algebras}. Explicitly, since transposition changes the order of composition from $\LH$ to $\LHop$, $C^{-1}(\rho)$ in Eq.~(\ref{eq: CI}) is more naturally identified with a map $C^{-1}(\rho): \LA^\mathrm{op} \ra \LB$ rather than with a map $\LA \ra \LB$---the latter being the natural identification for $J^{-1}(\rho)$ (see also Eq.~(\ref{eq: CJ vs JC}) below). Taking this into account will lead to a reformulation of Thm.~\ref{thm: Choi} in Thm.~\ref{thm: Choi revisited} below.

In order to distinguish between the different operator orderings in $\LH$ and $\LHop$ in a basis-independent way, we replace the basis-dependent operation of transposition with the order-reversing involution on the complex structure in Eq.~(\ref{eq: complex structure on LH}), given by the Hermitian adjoint $*$ (cf. Sec.~\ref{sec: Jordan vs C*-algebras}). \mf{Note from above that $\LH^* \cong \LHop$.} 

We will need the following straightforward result, which is implicit in Ref.~\cite{FrembsDoering2022a,FrembsDoering2022b,Frembs2022a}.\footnote{Lm.~\ref{lm: complete positivity under *} becomes Lm.~4 in Ref.~\cite{Frembs2022a}, when expressed in terms of time orientations (see also Ref.~\cite{FrembsDoering2022a}).}

\begin{lemma}\label{lm: complete positivity under *}
    Let $\phi \in \LAmB$. Then $\phi: \LA \ra \LB$ is completely positive if and only if $\phi \circ *: \LAS \ra \LB$ is completely positive.
\end{lemma}

\begin{proof}
    By Stinespring's theorem \cite{Stinespring1955}, every completely positive map $\phi: \LA \ra \LB$ is of the form $\phi = v^*\Phi v$ for a linear map $v: \cH_B \ra \cK$ and $\Phi: \LA \ra \LK$ a $C^*$-algebra homomorphism. 
    It follows that $\Phi \circ *: \LAS \ra \LK$ is also a $C^*$-algebra homomorphism,
    \begin{equation*}
        (\Phi \circ *)(a \cdot_- a') = \Phi((a'a)^*) = \Phi(a^*a'^*) = \Phi(a^*)\Phi(a'^*) = (\Phi \circ *)(a)(\Phi \circ *)(a')\; .
    \end{equation*}
    Hence, $\phi \circ * = v^*(\Phi \circ *)v: \LAS \ra \LB$ is completely positive \cite{Stinespring1955}.
\end{proof}

Note that a similar result holds for $*$ replaced by any order-reversing involution on $\LA$, in particular, it also holds for $T$. We thus obtain the following reformulation of Thm.~\ref{thm: Choi}.

\begin{theorem}[Choi (reformulated) \cite{Choi1975}]\label{thm: Choi revisited}
    Let $\phi \in \LAmB$. Then $\phi^* = \phi \circ *: \LAS \ra \LB$ is completely positive if and only if $\rho^J_{\phi^*} \in \mc{L}(\cH_A \otimes \cH_B)$ 
    is positive. That is, $J$ restricts to a map $J|_+: (\LAmB_\mathrm{CP} \circ *) \ra \LAxB_+$.
\end{theorem}

\begin{proof}
    By Eq.~(\ref{eq: JI}), $\rho_{\phi^*}^J = \sum_k e^*_k \otimes \phi^*(e_k)$, where $\{e_k\}_k \in \mathrm{ONB}(\LA)$ is any orthonormal basis in $\LA$. From this we compute
    \begin{align*}
        \rho^J_{\phi^*}
        = \sum_k e^*_k \otimes \phi^*(e_k)
        = \sum_{ij} |i\rangle\langle j| \otimes (\phi \circ *)(|j\rangle\langle i|)
        = \sum_{ij} |i\rangle\langle j| \otimes \phi(|i\rangle\langle j|)
        = \rho^C_\phi\; ,
    \end{align*}
    where we fixed the choice of basis to $\{e_k\}_k = \{|i\rangle\langle j|\}_{ij}$ in the second step. By Choi's theorem (Thm.~\ref{thm: Choi}), $\rho^J_{\phi^*} = \rho^C_\phi \in \mc{L}(\cH_A\otimes\cH_B)$ is positive if and only if $\phi: \LA \ra \LB$ is completely positive, if and only if $\phi^*: \LAS \ra \LB$ is completely positive by Lm.~\ref{lm: complete positivity under *}.
\end{proof}

\mf{Note that Thm.~\ref{thm: Choi revisited} complements the results in Ref.~\cite{PaulsenShultz2013,Kye2022}: since the action on itself by (unitary) conjugation preserves the order of composition (cf. Eq.~(\ref{eq: canonical time orientation})), Thm.~\ref{thm: Choi revisited} explains the invariance of Choi's theorem under the basis change in Ref.~\cite{PaulsenShultz2013,Kye2022}.}

We emphasise that---unlike Hermiticity-preserving, positive and decomposable maps, which are maps between Jordan algebras (see Thm.~\ref{thm: dePillis for Jordan algebras} and Thm.~\ref{thm: Jamiolkowski for Jordan algebras})---completely positive maps are maps between $C^*$-algebras: a completely positive map $\phi$ and its corresponding operator $\rho_\phi^{J}$ depend on the relative operator ordering between the respective algebras.\footnote{Equivalently, $\rho_\phi^J$ depends on the (relative) time orientation between the respective algebras (cf. \cite{FrembsDoering2022a,Frembs2022a}).} Since quantum states correspond with completely positive maps under Choi's theorem, they too are generally sensitive to the relative operator ordering in subalgebras \cite{FrembsDoering2022a}.\footnote{In fact, by Ref.~\cite{Frembs2022a} a quantum state is sensitive to the relative order if and only if it is entangled.}

\subsection{Variations on the Choi-Jamio\l kowski isomorphism}\label{sec: variations on the CJI}

With Thm.~\ref{thm: Choi revisited}, we define the following version of the Choi-Jamio\l kowski isomorphism:
\begin{align}\label{eq: CJI*}
    \begin{split}
    &(\rho^\CJ_\phi,a^*\otimes b)_{\LA\otimes\LB} := (\phi_{\LAS \ra \LB}(a),b)_\LB \\
    \Longleftrightarrow \hspace{1.035cm}
    &(\rho^\CJ_\phi,a^*\otimes b)_{\LA\otimes\LB}
    := (\phi^*_{\LA \ra \LB}(a^*),b)_\LB \\
    \Longleftrightarrow \hspace{1.035cm}
    &(\rho^\CJ_\phi,a\otimes b)_{\LA\otimes\LB}
    \hspace{0.18cm} := (\phi^*_{\LA \ra \LB}(a),b)_\LB
    \end{split}
\end{align}
for all $a \in \LA$, $b \in \LB$, where (for added emphasis) we indicate the signature of a map $\phi: \LAS \ra \LB$ in subscript $\phi_{\LAS \ra \LB}$. Clearly, the isomorphism is basis-independent since Eq.~(\ref{eq: CJI*}) is expressed in terms of inner products (similar to $J$ in Eq.~(\ref{eq: CJI})), and Thm.~\ref{thm: dePillis for Jordan algebras} and Thm.~\ref{thm: Jamiolkowski for Jordan algebras} apply since $\CJ$ agrees with $J$ under restriction to Jordan algebras. Finally, by comparison with Thm.~\ref{thm: Choi revisited}, $\rho^\CJ_\phi \in \LA\otimes\LB \cong \mc{L}(\cH_A\otimes\cH_B)$ is positive if and only if $\phi: \LAS \ra \LB$ is completely positive, if and only if $\phi^*: \LA \ra \LB$ is completely positive by Lm.~\ref{lm: complete positivity under *}. Eq.~(\ref{eq: CJI*}) therefore fixes a choice of operator ordering on the Jordan algebras $\JAsa$ and $\JBsa$; by analogy with the order-reversal imposed by transposition in Eq.~(\ref{eq: CI}), we denote the isomorphism in Eq.~(\ref{eq: CJI*}) by $\CJ$.\footnote{We remark that this version of the isomorphism is implicit in Ref.~\cite{FrembsDoering2022a,Frembs2022a}.}

This is not the only possible choice of operator orderings on $\JAsa$ and $\JBsa$. In particular, note that---different from Eq.~(\ref{eq: CJI*})---$\phi$ in Eq.~(\ref{eq: CJI}) is implicitly interpreted as a map $\phi: \LA \ra \LB$. The latter signature arises from Eq.~(\ref{eq: CJI*}) under a change of \mf{(time orientation in)} $C^*$-algebras $\LA \leftrightarrow \LAS \mf{\cong \LA^*}$: for all $a \in \LA$, $b \in \LB$,
\begin{align}\label{eq: CJ vs JC}
    \begin{split}
    &(\rho^\CJ_\phi,a\otimes b)_{\LA\otimes\LB}
    \hspace{0.575cm} = (\phi^*_{\LA \ra \LB}(a),b)_\LB \\
    \stackrel{\scalebox{0.6}{$\LA \leftrightarrow \LAS$}}{\longleftrightarrow} \quad
    &(\rho^\JC_\phi,a^*\otimes b)_{\LAS\otimes\LB} \hspace{0.03cm}
    := (\phi^*_{\LAS \ra \LB}(a^*),b)_\LB \\
    \Longleftrightarrow \hspace{1.035cm}
    &(\rho^\JC_\phi,a^*\otimes b)_{\LAS\otimes\LB} \hspace{0.03cm}
    := (\phi_{\LA \ra \LB}(a),b)_\LB\; .
    \end{split}
\end{align}
Observe that $\JC$ arises from $J$ under a change of identification from $\rho_\phi^J \in \mc{L}(\cH_A\otimes\cH_B)$ to $\rho^\JC_\phi \in \mc{L}(\cH^*_A\otimes\cH_B)$. \mf{Indeed, using Eq.~(\ref{eq: inner product in LH*}) between the first and the second line yields $\rho_\phi^\JC = (\rho^\CJ_\phi)^{*_A}$, where $*_A$ denotes the Hermitian adjoint on $\LA$.}
We thus denote the isomorphism in Eq.~(\ref{eq: CJ vs JC}) 
by $\JC$. Similar to $\CJ$, we find that $\JC$ is basis-independent and Thm.~\ref{thm: dePillis for Jordan algebras} and Thm.~\ref{thm: Jamiolkowski for Jordan algebras} apply. Moreover, \mf{by comparison with Thm.~\ref{thm: Choi revisited}, $\phi: \LA \ra \LB$ is completely positive if and only if $\rho^\JC_\phi \in \LAS\otimes\LB \cong \mc{L}(\cH^*_A\otimes\cH_B)$ is positive.} In this way, Eq.~(\ref{eq: CJ vs JC}) recovers Eq.~(\ref{eq: CJI* pedestrian}); in particular, the identification of the dual Hilbert space in Eq.~(\ref{eq: CJI* pedestrian}) merely reflects the fact that $\LAS$ naturally acts on $\cH^*_A$
rather than on $\cH_A$.

Finally, 
note that $\phi: \LA \ra \LB$ is completely positive if and only if $\phi: \LAS \ra \LBS$ is completely positive; similarly, $\phi: \LAS \ra \LB$ is completely positive if and only if $\phi: \LA \ra \LBS$ is completely positive.
It follows that 
Eq.~(\ref{eq: CJI*}) and Eq.~(\ref{eq: CJ vs JC}) are invariant under an overall change: $\LA \leftrightarrow \LAS$ and $\LB \leftrightarrow \LBS$, which corresponds to the symmetry of those equations under complex conjugation (cf. Eq.~(\ref{eq: inner product in LH*})).
In other words, $\rho^\CJ_\phi = (\rho^\CJ_\phi)^*$ is a state on both $\LA \otimes \LB$ and $\LAS \otimes \LBS$; similarly, $\rho^\JC_\phi = (\rho^\JC_\phi)^*$ is a state on both $\LAS \otimes \LB$ and $\LA \otimes \LBS$. Consequently, positivity of $\rho_\phi^\CJ$ (similarly $\rho_\phi^\JC$) encodes a \emph{relative choice of operator ordering between $\JAsa$ and $\JBsa$}. In summary, we record our \emph{variations on the Choi-Jamio\l kowski isomorphism}: for all $a \in \LA$, $b \in \LB$, \vspace{0.2cm}
\begin{align}\label{eq: variations on the CJI*}
    \begin{split}
    &(\rho^\CJ_\phi,a\otimes b)_{\LAS\otimes\LBS} \hspace{0.075cm} = (\phi^*_{\LAS \ra \LBS}(a),b)_{\LBS} \\ 
    \stackrel{\scalebox{0.6}{$\LA \leftrightarrow \LAS$}}{\stackrel{\scalebox{0.6}{$\LB \leftrightarrow \LBS$}}{\longleftrightarrow}} \quad &(\rho^\CJ_\phi,a\otimes b)_{\LA\otimes\LB} \hspace{0.675cm} = (\phi^*_{\LA \ra \LB}(a),b)_\LB \\
    \stackrel{\scalebox{0.6}{\color{white}{$\LA \leftrightarrow \LAS$}}}{\stackrel{\scalebox{0.6}{$\LA \leftrightarrow \LAS$}}{\longleftrightarrow}} \quad &(\rho^\JC_\phi,a^*\otimes b)_{\LAS\otimes\LB} \hspace{0.23cm} = (\phi_{\LA \ra \LB}(a),b)_\LB \\
    \stackrel{\scalebox{0.6}{$\LA \leftrightarrow \LAS$}}{\stackrel{\scalebox{0.6}{$\LB \leftrightarrow \LBS$}}{\longleftrightarrow}} \quad &
    (\rho^\JC_\phi,a^*\otimes b)_{\LA\otimes\LBS} \hspace{0.23cm} = (\phi_{\LAS \ra \LBS}(a),b)_{\LBS}\; .
    \end{split}
\end{align}

\section{Conclusion}

We have identified the physical significance of the basis-dependent definition of the Choi matrix ($\rho^C_\phi$ in Eq.~(\ref{eq: CI})) in Choi's theorem (Thm.~\ref{thm: Choi}) with the relative operator ordering between the $C^*$-algebras describing the two subsystems (Thm.~\ref{thm: Choi revisited}). Moreover, different choices of relative operator orderings distinguish between different versions of the Choi-Jamio\l kowski ismorphism (Eq.~(\ref{eq: variations on the CJI*})). Our analysis is closely related with the intrinsic dynamics in subsystems, as described by their respective time orientations \cite{Doering2014,FrembsDoering2022a,FrembsDoering2022b,Frembs2022a}. 
The latter play a crucial role for the classification of bipartite quantum states from more general non-signalling correlations \cite{FrembsDoering2022a}, and further allow for a classification of bipartite entanglement \cite{Frembs2022a}. 

Since the Choi-Jamio\l kowski isomorphism is a basic tool across all of quantum information theory, we expect the techniques used here will also shed light on various other results in the field. For instance, due to its importance in the framework of process matrices \cite{OreshkovCostaBrukner2012,CostaShrapnel2016} and quantum causal models \cite{LeiferSpekkens2013,CostaShrapnel2016,AllenEtAl2017,BarrettLorenzOreshkov2020}, we suspect that an analysis which parallels the one given here will prove useful in the further advancement of the subject.

\paragraph*{Acknowledgements.} MF thanks Francesco Buscemi and Andreas D\"oring for discussions. The authors acknowledge financial support through grant number FQXi-RFP-1807 from the Foundational Questions Institute and Fetzer Franklin Fund, a donor advised fund of Silicon Valley Community Foundation, and ARC Future Fellowship FT180100317.

\bibliographystyle{siam}
\bibliography{bibliography}

\end{document}